\documentclass[conference, 10pt]{IEEEtran}

\usepackage[cmex10]{amsmath}  
\usepackage{amssymb}
\usepackage{graphicx,color}   

\usepackage{ifxetex}
\ifxetex
\usepackage{xunicode}
\else
\usepackage[utf8]{inputenc}
\usepackage[T1]{fontenc}
\usepackage{microtype}
\fi

\usepackage{cite}  
\usepackage[caption=false,font=footnotesize]{subfig} 
\usepackage[boxed]{algorithm2e}

\newtheorem{theorem}{Theorem}
\newtheorem{remark}{Remark}
\newtheorem{lemma}{Lemma}
\newtheorem{definition}{Definition}

\long\def\symbolfootnote[#1]#2{\begingroup%
\def\thefootnote{\fnsymbol{footnote}}\footnote[#1]{#2}\endgroup}

\newcommand{\od}{\mathsf{od}}
\newcommand{\maxod}{\mathsf{od}_{\text{max}}}
\newcommand{\gap}{\vspace{-0ex}}

\begin{document}

\title{Optimal Index Codes for a Class of Multicast Networks with Receiver Side Information}
\author{\IEEEauthorblockN{Lawrence Ong}
\IEEEauthorblockA{School of Electrical Engineering and Computer Science,\\The University of Newcastle, Australia \\
Email: lawrence.ong@cantab.net}
\and
\IEEEauthorblockN{Chin Keong Ho}
\IEEEauthorblockA{Institute for Infocomm Research, A*STAR,\\ 1 Fusionopolis Way, \#21-01 Connexis, Singapore \\
Email: hock@i2r.a-star.edu.sg}
}
\maketitle

\begin{abstract}
This paper studies a special class of multicast index coding problems where a sender transmits messages to multiple receivers, each with some side information. 
Here, each receiver knows a unique message a priori, and there is no restriction on how many messages each receiver requests from the sender. For this class of multicast index coding problems, we obtain the optimal index code, which has the shortest codelength for which the sender needs to send in order for all receivers to obtain their (respective) requested messages. This is the first class of index coding problems where the optimal index codes are found. In addition, linear index codes are shown to be optimal for this class of index coding problems.
\end{abstract}

\section{Introduction}

\subsection{Background}

Consider the following communication problem of broadcasting with receiver side information. A single sender wishes to send a set of $m$ messages $\mathcal{M} = \{x_1, x_2, \dotsc, x_m\}$ to a set of $n$ receivers $\mathcal{R}= \{R_1,R_2\dotsc,R_n\}$. Each receiver is defined as $R_i \triangleq (\mathcal{W}_i,\mathcal{K}_i)$, i.e., it knows some messages $\mathcal{K}_i \subseteq \mathcal{M}$ a priori, and it wants to obtain some messages $\mathcal{W}_i \subseteq \mathcal{M}$. This is known as the {\em index coding} problem~\cite{elrouayheb10}, and any index coding problem can be completely specified by $(\mathcal{M},\mathcal{R})$. In this paper, we consider only binary messages\footnote{The results here also apply to messages of equal size that are non-binary.}, i.e., $x_i \in \{0,1\}$ for all $i \in \{1,2,\dotsc,m\}$, where $x_i$ are each uniformly distributed on $\{0,1\}$ and are mutually independent.

An {\em index code} for the index coding problem is defined as:
\begin{definition}[Index Code]
An index code for the index coding problem $(\mathcal{M},\mathcal{R})$ consists of
\begin{enumerate}
\item An encoding function for the sender, $E: \{0,1\}^m \rightarrow \{0,1\}^\ell$, and
\item A decoding function for each receiver, $D_i: \{0,1\}^{\ell+|\mathcal{K}_i|} \rightarrow \{0,1\}^{|\mathcal{W}_i|}$ such that $D_i(E(\mathcal{M}),\mathcal{K}_i) = \mathcal{W}_i$, for each $i \in \{1,2,\dotsc, n\}$.
\end{enumerate}
\end{definition}
In words, the sender encodes its $m$-bit message word into an $\ell$-bit codeword which is given to all receivers. Using the codeword and its known messages, each receiver decodes the messages that it wants.
The integer $\ell$ is referred to as the {\em length} of the index code. Let $\ell^*(\mathcal{M},\mathcal{R})$ be the smallest integer $\ell$ for which the above conditions hold. 

Our objective is to determine $\ell^*(\mathcal{M},\mathcal{R})$ and construct an {\em optimal} index code that has length $\ell^*(\mathcal{M},\mathcal{R})$. In practice, this leads to the optimal use of transmission energy and resources.

Without loss of generality, we assume that $|\mathcal{W}_i| \geq 1$ and $|\mathcal{K}_i| \geq 1$ for each $i \in \{1,2,\dotsc, n\}$, meaning that each receiver knows at least one bit and requests for at least one bit. This is because (i) any receiver that does not requests for any message bit can be removed from the system, and so we do not need to consider  the case where $\mathcal{W}_i=0$, and (ii) if a receiver $i$ knows no message bit, we can arbitrarily assign a new dummy bit $x'$ to it and to the sender (of course, that bit will never be sent by the sender), and so we do not consider the case where $\mathcal{K}_i=0$.

\subsection{Classification}

A few classes of index coding problems have been studied. We propose to categorize these and other index coding problems as follows. We first classify different types of information flow from the sender to the receivers. We say that an index coding problem is {\em unicast} if
\gap
\begin{equation}
\mathcal{W}_i \cap \mathcal{W}_j = \emptyset, \quad \forall i \neq j, \label{eq:unicast}
\end{equation}
meaning that each message bit can be requested by at most one receiver.
 In addition, we say that the problem is {\em single-unicast} if, in addition to \eqref{eq:unicast}, we also have that $|\mathcal{W}_i| = 1$ for all $i$, meaning each receiver request for exactly one unique bit.

We next classify different types of side information at the receivers. We say that an index coding problem is {\em uniprior} if
\gap
\begin{equation}
\mathcal{K}_i \cap \mathcal{K}_j = \emptyset, \quad \forall i \neq j, \label{eq:uniprior}
\end{equation}
meaning that each bit is known a priori to at most one receiver. In addition, we say that the problem is {\em single-uniprior} if, in addition to \eqref{eq:uniprior}, we also have that $|\mathcal{K}_i| = 1$ for all $i$, meaning that each receiver knows exactly one unique bit a priori.

With the above terminology, we canonically term the general index coding problems (i.e., no restriction on all $\mathcal{W}_i$ and $\mathcal{K}_i$)  {\em multicast/multiprior} problems.

\subsection{Different Classes of Index Coding Problems}

\subsubsection{Single-Unicast/Multiprior (or simple Single-Unicast)}

Birk and Kol~\cite{birkkol2006} and Bar-Yossef et al.~\cite{baryossefbirk11} studied index coding problems with single-unicast and multiprior.
In this setting, each receiver wants one unique message bit, and there is no restriction on how many messages the receivers know a priori.
So, we have that $R_i = (x_i, \mathcal{K}_i)$ for all $i$, and $m=n$. Bar-Yossef et al. represent single-unicast problems by {\em side-information graphs} with $n$ vertices $\{1,2,\dotsc,n\}$, where an edge exists from vertex $i$ to vertex $j$ if and only if receiver $i$ knows $x_j$ a priori, i.e., $x_j \in \mathcal{K}_i$.

\subsubsection{Single-Uniprior/Multicast (or simply Single-Uniprior)}

In this paper, we consider index coding problems with single-uniprior and multicast, where each receiver knows only one message bit (in contrast to the above single-unicast problems where each receiver wants one message), and there is no restriction on how many messages each receiver wants. Here, we have that $R_i = (\mathcal{W}_i, x_i)$ for all $i$, and $m=n$. Side-information graphs used to represent single-unicast problems cannot capture all single-uniprior problems. So, we represent single-uniprior problems with {\em information-flow graphs} of $n$ vertices, but an arc exists from vertex $i$ to vertex $j$ if and only if node $j$ wants $x_i$, i.e., $x_i \in \mathcal{W}_j$. The single-uniprior problem is motivated by the bidirectional relaying network, which will be discussed later.

\begin{figure}[t]
\centering
\resizebox{3cm}{!}{ 
\begin{picture}(0,0)%
\includegraphics{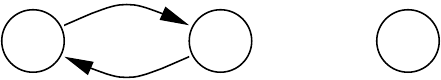}%
\end{picture}%
\setlength{\unitlength}{3947sp}%
\begingroup\makeatletter\ifx\SetFigFont\undefined%
\gdef\SetFigFont#1#2#3#4#5{%
  \fontsize{#1}{#2pt}%
  \fontfamily{#3}\fontseries{#4}\fontshape{#5}%
  \selectfont}%
\fi\endgroup%
\begin{picture}(2116,374)(743,-1148)
\put(2626,-1036){\makebox(0,0)[lb]{\smash{{\SetFigFont{12}{14.4}{\familydefault}{\mddefault}{\updefault}{\color[rgb]{0,0,0}$3$}%
}}}}
\put(826,-1036){\makebox(0,0)[lb]{\smash{{\SetFigFont{12}{14.4}{\familydefault}{\mddefault}{\updefault}{\color[rgb]{0,0,0}$1$}%
}}}}
\put(1726,-1036){\makebox(0,0)[lb]{\smash{{\SetFigFont{12}{14.4}{\familydefault}{\mddefault}{\updefault}{\color[rgb]{0,0,0}$2$}%
}}}}
\end{picture}%

}
\caption{A directed graph representing index coding problems}
\label{fig:graph}
\end{figure}
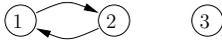

\subsubsection{Side-information Graphs versus Information-Flow Graphs}

Consider the directed graph in Fig.~\ref{fig:graph}. On the one hand, if it is a side-information graph, we have the following single-unicast problem: $R_1 = (x_1,x_2)$, $R_2=(x_2,x_1)$, and $R_3=(x_3, x')$, where $x'$ is a dummy bit known to receiver 3 and the sender. An optimal index code is $(x_1 \oplus x_2, x_3)$, i.e. $\ell(^*\mathcal{M}, \mathcal{R}) = 2$. On the other hand, if the graph is an information-flow graph, we have the following single-uniprior problem: $R_1 = (x_2,x_1)$, $R_2=(x_1,x_2)$, and $R_3=(\emptyset, x_3)$. In this case, receiver 3 can be removed. An optimal index code is $(x_1 \oplus x_2)$, i.e.,  $\ell^*(\mathcal{M}, \mathcal{R}) = 1$. We note that designating a given directed graph as a side-information graph or an information-flow graph can lead to different index coding problems and hence possibly different optimal codes.


\subsubsection{Unicast/Uniprior}

The class of index coding problems with unicast and uniprior was investigated by Neely et al.~\cite{neely11}, where (i) each message bit is known to only one receiver, and (ii) each message bit is requested by only one receiver. However, there is no restriction on the number of message bits each receiver knows or requests. Unicast/uniprior problems can be represented by (modified) information-flow graphs.

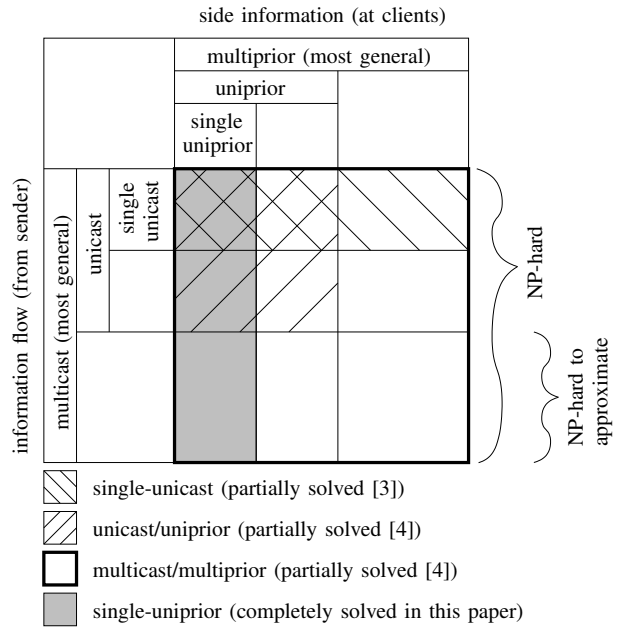
\begin{figure}[t]
\centering
\resizebox{0.92\linewidth}{!}{ 
\begin{picture}(0,0)%
\includegraphics{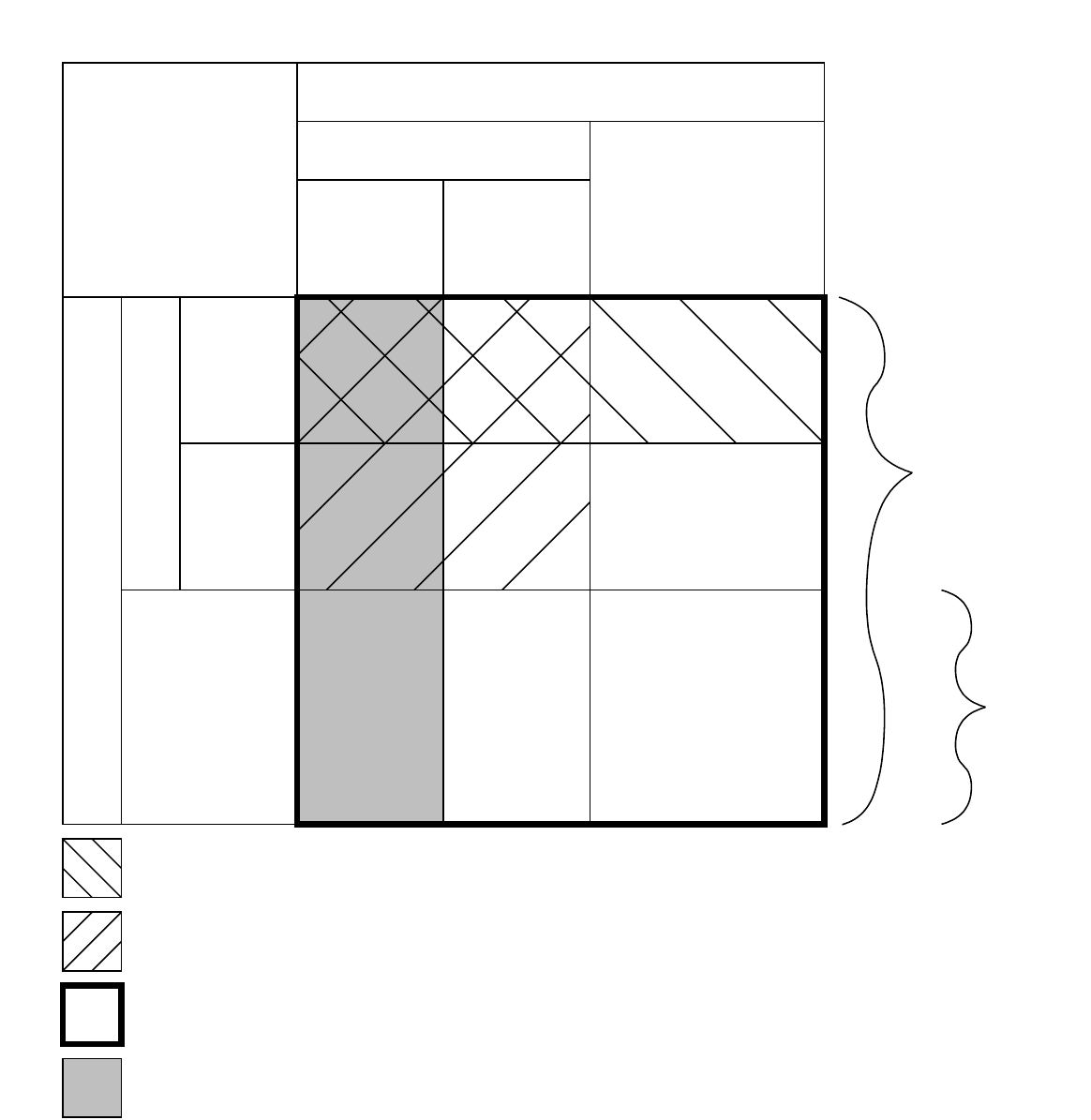}%
\end{picture}%
\setlength{\unitlength}{3947sp}%
\begingroup\makeatletter\ifx\SetFigFont\undefined%
\gdef\SetFigFont#1#2#3#4#5{%
  \reset@font\fontsize{#1}{#2pt}%
  \fontfamily{#3}\fontseries{#4}\fontshape{#5}%
  \selectfont}%
\fi\endgroup%
\begin{picture}(5570,5733)(729,-5323)
\put(1276,-3511){\rotatebox{90.0}{\makebox(0,0)[lb]{\smash{{\SetFigFont{12}{14.4}{\familydefault}{\mddefault}{\updefault}{\color[rgb]{0,0,0}multicast  (most general)}%
}}}}}
\put(2551,-136){\makebox(0,0)[lb]{\smash{{\SetFigFont{12}{14.4}{\familydefault}{\mddefault}{\updefault}{\color[rgb]{0,0,0}multiprior  (most general)}%
}}}}
\put(1576,-2086){\rotatebox{90.0}{\makebox(0,0)[lb]{\smash{{\SetFigFont{12}{14.4}{\familydefault}{\mddefault}{\updefault}{\color[rgb]{0,0,0}unicast}%
}}}}}
\put(2626,-436){\makebox(0,0)[lb]{\smash{{\SetFigFont{12}{14.4}{\familydefault}{\mddefault}{\updefault}{\color[rgb]{0,0,0}uniprior}%
}}}}
\put(2101,-1786){\rotatebox{90.0}{\makebox(0,0)[lb]{\smash{{\SetFigFont{12}{14.4}{\familydefault}{\mddefault}{\updefault}{\color[rgb]{0,0,0}unicast}%
}}}}}
\put(2326,-961){\makebox(0,0)[lb]{\smash{{\SetFigFont{12}{14.4}{\familydefault}{\mddefault}{\updefault}{\color[rgb]{0,0,0}uniprior}%
}}}}
\put(1876,-1711){\rotatebox{90.0}{\makebox(0,0)[lb]{\smash{{\SetFigFont{12}{14.4}{\familydefault}{\mddefault}{\updefault}{\color[rgb]{0,0,0}single}%
}}}}}
\put(2401,-736){\makebox(0,0)[lb]{\smash{{\SetFigFont{12}{14.4}{\familydefault}{\mddefault}{\updefault}{\color[rgb]{0,0,0}single}%
}}}}
\put(1501,-4111){\makebox(0,0)[lb]{\smash{{\SetFigFont{12}{14.4}{\familydefault}{\mddefault}{\updefault}{\color[rgb]{0,0,0}single-unicast (partially solved~\cite{baryossefbirk11})}%
}}}}
\put(1501,-4486){\makebox(0,0)[lb]{\smash{{\SetFigFont{12}{14.4}{\familydefault}{\mddefault}{\updefault}{\color[rgb]{0,0,0}unicast/uniprior (partially solved~\cite{neely11})}%
}}}}
\put(2476,239){\makebox(0,0)[lb]{\smash{{\SetFigFont{12}{14.4}{\familydefault}{\mddefault}{\updefault}{\color[rgb]{0,0,0}side information (at clients)}%
}}}}
\put(901,-3736){\rotatebox{90.0}{\makebox(0,0)[lb]{\smash{{\SetFigFont{12}{14.4}{\familydefault}{\mddefault}{\updefault}{\color[rgb]{0,0,0}information flow (from sender)}%
}}}}}
\put(6226,-3661){\rotatebox{90.0}{\makebox(0,0)[lb]{\smash{{\SetFigFont{12}{14.4}{\familydefault}{\mddefault}{\updefault}{\color[rgb]{0,0,0}approximate}%
}}}}}
\put(1501,-4861){\makebox(0,0)[lb]{\smash{{\SetFigFont{12}{14.4}{\familydefault}{\mddefault}{\updefault}{\color[rgb]{0,0,0}multicast/multiprior (partially solved~\cite{neely11})}%
}}}}
\put(1501,-5236){\makebox(0,0)[lb]{\smash{{\SetFigFont{12}{14.4}{\familydefault}{\mddefault}{\updefault}{\color[rgb]{0,0,0}single-uniprior (completely solved in this paper)}%
}}}}
\put(6001,-3661){\rotatebox{90.0}{\makebox(0,0)[lb]{\smash{{\SetFigFont{12}{14.4}{\familydefault}{\mddefault}{\updefault}{\color[rgb]{0,0,0}NP-hard to}%
}}}}}
\put(5626,-2311){\rotatebox{90.0}{\makebox(0,0)[lb]{\smash{{\SetFigFont{12}{14.4}{\familydefault}{\mddefault}{\updefault}{\color[rgb]{0,0,0}NP-hard}%
}}}}}
\end{picture}%

}
\caption{Summary of main result and related results}
\label{fig:result}
\end{figure}

\subsection{Existing Results and New Results}

For single-unicast index coding problems, Bar-Yossef et al.~\cite{baryossefbirk11} found optimal index codes for problems that can be represented by the following types of side-information graphs: (i) acyclic graphs, (ii) perfect graphs, (iii) odd holes, and (iv) odd anti-holes. This means the corresponding classes of single-unicast problems were solved. In addition, {\em linear} index codes are optimal for these problems. Lubetzky and Stav~\cite{lubertzkystav09}, however, show that non-linear index codes can outperform linear codes for some single-unicast problems.

Neely et al.~\cite{neely11} solved the class of unicast/uniprior problems where the corresponding (modified) information-flow graphs have disjoint cycles. In addition, they found the optimal index codes for the general multicast/multiprior index coding problem where the corresponding bipartite graphs\footnote{Neither the side-information graph nor the information-flow graph is sufficient to represent the general index coding problem. So, bipartite graphs are used to represent multicast/multiprior index coding problems.} is acyclic.

It has been shown~\cite{chaudhry11} that (i) the general multicast/multiprior index coding problem is NP-hard, and (ii) the multicast (non-unicast) index problem is even NP-hard to approximate.

In this paper, we solve the single-uniprior (multicast) index coding problem, and show that the solution can be found in polynomial time. The result of this paper in relation to existing results is summarized in Fig.~\ref{fig:result}.

More specifically, in this paper, we construct an optimal index code for any single-uniprior index coding problem, which can be represented by an information-flow graph. 
In addition, the optimal index codes that we construct are linear. Hence, we incidentally show that linear index codes are optimal for all single-uniprior problems.

\begin{remark}
To the best of our knowledge, the single-uniprior problem is the only class of index coding problems where the solution is found for any configuration of information flow and side information. Furthermore, in contrast to the single-unicast problems where non-linear codes can outperform linear codes~\cite{lubertzkystav09}, we show that linear codes are optimal for the single-uniprior problems.
\end{remark}

\subsection{Motivating the Single-Uniprior Index Coding Problem}

The single-uniprior index coding problem formulation is motivated by satellite communications~\cite{gunduzyener09,ongmjohnsonit11}, where multiple clients exchange messages through a relay (i.e., the satellite). Each clients wish to send its data to a predefined set of other clients. As there is no direct communication link among the clients, they first send their respective messages to the relay on the {\em uplink} channel. The relay then processes the messages and broadcasts a re-encoded message to the clients on the {\em downlink} channel. Assuming that the relay has obtained all the messages from the clients on the uplink, we wish to know the minimum number of bits the relay needs to send on the downlink in order for each client to obtain it requested messages. This is exactly the single-uniprior index coding problem, where each receiver is interested in the messages from different receiver(s), and it only knows its own message a priori.

\subsection{Information-Flow Graph}

As mentioned, we represent single-uniprior problems using information-flow graphs. Let the set of all vertices be $\mathcal{V} = \{1,2,\dotsc,n\}$. An arc denoted by an ordered pair of vertices $a=(i,j)$ exists if receiver $j$ wants the message $x_i$, i.e., $x_i \in \mathcal{W}_j$. We call vertex $i$ the tail and $j$ the head of the arc $(i,j)$. We denote the graph, by $\mathcal{G}=(\mathcal{V}, \mathcal{A})$, where $\mathcal{A}$ is the set of all arcs.
Let the number of vertices and arcs be $V(\mathcal{G}) \triangleq |\mathcal{V}|$ and $A(\mathcal{G}) \triangleq |\mathcal{A}|$, respectively.

Any information-flow graph $\mathcal{G}$ has the following structure: (i) $\mathcal{G}$ is a {\em simple graph}, i.e., the head and tail of every edge are distinct (otherwise the receiver requests for the message it already knows). (ii) At most one arc exists between any ordered pair of vertices (single-uniprior and the messages are of the same size). (iii) Each vertex is connected to at least one other vertex via an arc (otherwise, we can remove the vertex [say vertex $i$], as $x_i$ is not requested by any receiver and receiver $i$ does not want any message). Condition (iii) implies that there are at least two vertices and at least one arc on $\mathcal{G}$, i.e., $V(\mathcal{G}) \geq 2$ and $A(\mathcal{G}) \geq 1$.

We list the relevant standard definitions in graph theory~\cite{Bang-JensenGutin}:
\begin{itemize}
\item For a vertex $i$, $a=(j,i)$ is an \emph{incoming arc}, and $a=(i,j)$ an \emph{outgoing arc}.
The number of outgoing arcs is the {\em outdegree}  $\od$ of the vertex.
The maximum number of outgoing arcs of a graph is its {\em maximum outdegree} $\maxod$.

\item A \emph{trail} $\mathcal{T}$ is a non-empty graph of the form
$( \{ k_1,k_2,\dotsc,k_K\}, \{ (k_1,k_2), (k_2,k_3), \dotsc, (k_{K-1},k_K)\})$,
where arcs $(k_i,k_{i+1})$ are all distinct. Vertex $k_1$ is the tail of the trail, $k_K$ the head, and $\{k_2, k_3, \dotsc, k_{K-1}\}$ the {\em inner vertices}.
\item A \emph{path} $\mathcal{P}$ is a trail where the vertices are all distinct.
The vertex $k_1$ is the tail of the path and $k_K$ is the head of the trail, and we say that the path is from $k_1$ to $k_K$.
\item A \emph{cycle} is a path with an additional arc $(k_K,k_1)$.
\item A graph is {\em acyclic} if it does not contain any cycle.
\item A graph $\mathcal{G}'=(\mathcal{V}',\mathcal{A}')$ is a subgraph of a graph $\mathcal{G}=(\mathcal{V},\mathcal{A})$, denoted $\mathcal{G}' \subseteq \mathcal{G}$,  if $\mathcal{V}' \subseteq \mathcal{V}$ and $\mathcal{A}' \subseteq \mathcal{A}$. We say that the subgraph $\mathcal{G}'$ is \emph{on} the graph $\mathcal{G}$. Moreover, $\mathcal{G}'$ is a {\em strict subgraph} if $\mathcal{V}' \subset \mathcal{V}$ or $\mathcal{A}' \subset \mathcal{A}$.
\item A graph $\mathcal{V}$ is {\em strongly connected} 
if there exists a path for every distinct ordered pair of receivers.
\item The {\em strongly connected components} of a graph are its maximal strongly connected subgraphs.
\item For two graphs $\mathcal{G}_1 = (\mathcal{V}_1, \mathcal{A}_1)$ and $\mathcal{G}_2 = (\mathcal{V}_2, \mathcal{A}_2)$, we define: $\mathcal{G}_1 \cup \mathcal{G}_2 \triangleq ( \mathcal{V}_1 \cup \mathcal{V}_2, \mathcal{A}_1 \cup \mathcal{A}_2)$, $\mathcal{G}_1 \cap \mathcal{G}_2 \triangleq ( \mathcal{V}_1 \cap \mathcal{V}_2, \mathcal{A}_1 \cap \mathcal{A}_2)$, and $\mathcal{G}_1 \setminus \mathcal{G}_2 \triangleq (\mathcal{V}_1 \setminus \mathcal{V}_2, \mathcal{A}_1 \setminus \mathcal{A}_2)$.
\end{itemize}

\subsection{Main Idea} \label{sec:graph-compare}

Since there is a one-to-one mapping between a single-uniprior index coding problem $(\mathcal{M},\mathcal{R})$ and its corresponding information-flow graph $\mathcal{G}$, we define $\ell^*(\mathcal{G}) \triangleq \ell^*(\mathcal{M},\mathcal{R})$.

\begin{lemma}\label{lem:basic}
Let $\ell$ be the length of an index code for the single-uniprior problem represented by $\mathcal{G}$.
If 
$\mathcal{G}' = (\mathcal{V}, \mathcal{A}') \subseteq \mathcal{G} = (\mathcal{V}, \mathcal{A})$, i.e., $\mathcal{G}'$ and $\mathcal{G}$ have the same vertices, but $\mathcal{A}' \subseteq \mathcal{A}$,
then $\ell^* (\mathcal{G}') \leq  \ell^* (\mathcal{G}) \leq \ell$.
\end{lemma}
\begin{proof}
By definition, any index code must satisfy $\ell \geq \ell^* (\mathcal{G})$.
If we add additional decoding requirements at the receivers (i.e., arcs on graphs), the sender cannot transmit fewer bits, i.e, $\ell^* (\mathcal{G}') \leq  \ell^* (\mathcal{G})$. 
\end{proof}

The main idea to prove the results in this paper is captured in Lemma~\ref{lem:basic}. We find a lower bound for $\ell^*(\mathcal{G})$ by choosing an appropriate $\mathcal{G}'$ (using our proposed {\em pruning algorithm}) where $\ell^*(\mathcal{G}')$ can be easily obtained. We then show that we can always construct an index code for the {\em original graph} $\mathcal{G}$ with $\ell  = \ell^* (\mathcal{G}') $. With this, we establish $\ell^* (\mathcal{G})$. Note that constructing an index code for $\mathcal{G}'$ is insufficient here---we need to construct an index code for $\mathcal{G}$ although the lower bound is obtained based on $\mathcal{G}'$.

We first consider two special classes of graphs in Section~\ref{section:free} and Section~\ref{section:strong}, which are used as building blocks for the general results (arbitrary graphs) in Section~\ref{section:general}.

\section{Acyclic Graphs with $\maxod=1$ } \label{section:free}

As mentioned above, for any graph $\mathcal{G}$, we will first prune $\mathcal{G}$ to get $\mathcal{G}'$ for which $\ell^*(\mathcal{G}')$ can be obtained easily. More specifically, $\mathcal{G}'$ is an acyclic graph with $\maxod=1$. In this section, we establish the optimal codelength for any acyclic graph $\mathcal{G}'$ with $\maxod=1$, i.e., $\ell^*(\mathcal{G}')$. In the subsequent sections, we will then show how to choose an appropriate $\mathcal{G}'$ such that we can construct an index code for $\mathcal{G}$ with $\ell = \ell^*(\mathcal{G}')$.

In this section, we will prove the following theorem:
\begin{theorem} \label{theorem:1}
For a single-uniprior problem represented by an acyclic graph $\mathcal{G}$ with $\maxod=1$, we have
\gap
\begin{equation}
\ell^*(\mathcal{G}) = A(\mathcal{G}). \label{eq:ub}
\end{equation}
\end{theorem}

The proof of Theorem~\ref{theorem:1} can be found in the appendix. 

\begin{remark}\label{rem:1}
The arcs in $\mathcal{G}$ represent all the message bits requested by the receivers.
Since there are only $A(\mathcal{G})$ unique bits requested, the sender simply sends these bits in their entirety, i.e., uncoded using time-division multiple-access (TDMA).
\end{remark}

\section{Strongly Connected Graphs} \label{section:strong}

Next, we consider strongly connected graphs. We will show the following:
\begin{theorem}\label{theorem:strong}
For a single-uniprior problem represented by a strongly connected graph $\mathcal{G}$, we have
\gap
\begin{equation}
\ell^*(\mathcal{G}) = V(\mathcal{G})-1. \label{eq:strong-capacity}
\end{equation}
\end{theorem}

\begin{remark}\label{rem:2}
Consider the index coding problem represented by a  strongly connected graph. Since every vertex must have an outgoing arc for the graph to be strongly connected, there are in total $V(\mathcal{G})$ unique bits requested by the receivers. If we use the uncoded TDMA scheme, then the receiver needs to transmit $V(\mathcal{G})$ bits, which is strictly sub-optimal. 
\end{remark}

We present the proof of Theorem~\ref{theorem:strong} in the following two subsections.
While the coding scheme is relatively simple (using network coding), the challenge is to show that the sender cannot send less than  $(V(\mathcal{G})-1)$ bits. 

\subsection{Achievability (Upper Bound)} \label{sec:achievability}
We now propose a coding scheme that achieves $\ell = V(\mathcal{G})-1$. Recall that the set of vertices $\mathcal{V}=\{1,2,\dotsc,V(\mathcal{G})\}$ on $\mathcal{G}$ represent the receivers of the uniprior problem. Define $x_{i,j} \triangleq x_i \oplus x_j$, where $\oplus$ is the XOR operation. Now, let $\boldsymbol{x} = \big( x_{1,2}, x_{2,3}, \dotsc, x_{V(\mathcal{G})-2, V(\mathcal{G})-1}, x_{V(\mathcal{G})-1, V(\mathcal{G})} \big)$, which is a binary vector of length $(V(\mathcal{G})-1)$. The sender broadcasts $\boldsymbol{x}$. Note that each receiver $i$ knows $x_i$ a priori, for all $i \in \{1,2,\dotsc, V(\mathcal{G})\}$. Together with $\boldsymbol{x}$ received from the sender, receiver $i$ can decode all $\big\{x_j: j \in \{1,2,\dotsc, V(\mathcal{G})\} \setminus \{i\} \big\}$. So, we have $\ell^*(\mathcal{G}) \leq \ell = V(\mathcal{G})-1$. \hfill $\blacksquare$


\begin{remark}
This coding scheme also allows each receiver to decode all the message bits, besides the bit(s) it requested.
\end{remark}

\subsection{Lower Bound}

To obtain a lower bound on $\ell^*(\mathcal{G})$, we will construct an algorithm that prunes some arcs from $\mathcal{G}$ to obtain an acyclic graph with $\maxod=1$, say $\mathcal{G}''$, such that  $A(\mathcal{G}'') = V(\mathcal{G})-1$. From Lemma~\ref{lem:basic} we have that $\ell^*(\mathcal{G}) \geq \ell^*(\mathcal{G}'')$.  From Theorem~\ref{theorem:1} that applies to $\mathcal{G}''$, we have that  $\ell^*(\mathcal{G}'') =A(\mathcal{G}'')$. Hence, $\ell^*(\mathcal{G}) \geq  V(\mathcal{G})-1$.

\subsubsection{Graph Construction}

We start with a way to construct any strongly connected graph $\mathcal{G} = (\mathcal{V},\mathcal{A})$. In a strongly connected graph, there is a path from any vertex to another vertex (and vice versa). Thus $\mathcal{G}$ must contain at least one cycle.

\underline{Strongly Connected Graph Construction (SCGC):}
\begin{enumerate}
\item \label{step:add-path-1} Initialization: pick a cycle $\mathcal{C}$ on $\mathcal{G}$, and initialize $\mathcal{G}' = (\mathcal{V}',\mathcal{A}') \leftarrow \mathcal{C}$. 
\item \label{step:add-path-2} Iteration: pick a length $(K-1)$ trail on $\mathcal{G}$, denoted as $\mathcal{T}=( \{k_1, \dotsc, k_K\}, \{(k_1,k_2), \dotsc, (k_{K-1},k_K)\} )$,  to be either
\\
(i)  a path $\mathcal{P}$ where $K\geq 2$ and $k_1\neq k_K$, or
\\
(ii) a cycle $\mathcal{C}$ where $K\geq 3$ and $k_1=k_K$, \\
such that the tail and head satisfy $k_1, k_K\in \mathcal{V}'$ and the {\em inner vertices}, if present, are distinct and satisfy $k_i \in \mathcal{V}\setminus \mathcal{V'}, \forall i \in \{2,3,\dotsc, K-1\}$.
 We call $(k_1,k_2)$ the \emph{first arc} of $\mathcal{T}$.
The iteration terminates if such a trail cannot be found.

\item \label{step:add-path-3} Update: $\mathcal{G}' \leftarrow \mathcal{G}' \cup\mathcal{T}$.
We say the trail is {\em appended} to the graph. Go to Step~\ref{step:add-path-2}.
\end{enumerate}

\begin{lemma}\label{lem:strong}
Every iteration in Step~\ref{step:add-path-3} of the SCGC produces a strongly connected graph $\mathcal{G}'$ that is a subgraph of $\mathcal{G}$.
\end{lemma}
\begin{proof}
Assume that in Step~\ref{step:add-path-2}, $\mathcal{G}'$ is strongly connected and a subgraph of $\mathcal{G}$.
This is true for the first iteration, since $\mathcal{G}'$ is a cycle on $\mathcal{G}$ in Step~\ref{step:add-path-1}.
In Step~\ref{step:add-path-3}, $\mathcal{G}' \cup \mathcal{T}$ is a subgraph of  $\mathcal{G}$ because $\mathcal{T}$ is on  $\mathcal{G}$, and is also strongly connected because any vertex in $\mathcal{G}'$ can reach any vertex in the appended $\mathcal{T}$ via vertex $k_1$, or vice versa via vertex $k_K$.
By induction the properties hold for every iteration.
\end{proof}

\begin{lemma}
Any non-trivial strongly connected graph $\mathcal{G}$ can be generated with the SCGC, i.e., $\mathcal{G}'= \mathcal{G}$ after the SCGC terminates in Step~\ref{step:add-path-2}.
\end{lemma}

\begin{proof}
Step~\ref{step:add-path-1} is always possible, since any strongly connected graph must contain at least one cycle.

Suppose $\mathcal{G}'=\mathcal{G}$. Then it is not possible to find the trail $\mathcal{T}$ in Step~\ref{step:add-path-2}. Hence the iteration terminates with  $\mathcal{G}'= \mathcal{G}$.

Suppose $\mathcal{G}' = (\mathcal{V}',\mathcal{A}')$ is a {\em strict} subgraph of $\mathcal{G}=(\mathcal{V},\mathcal{A})$.
Denote $\mathcal{V}'^{\text{c}}=\mathcal{V}\setminus \mathcal{V'}, \mathcal{A}'^{\text{c}}=\mathcal{A}\setminus \mathcal{A'}$.
To complete the proof, we show that the trail in Step~\ref{step:add-path-2} can always be found, so that the iteration can continue until the algorithm terminates.

Without loss of generality, let $\mathcal{V}' \subseteq \mathcal{V}$ and $\mathcal{A}' \subset \mathcal{A}$,  i.e, $\mathcal{A}'^{\text{c}}$ is non-empty.
Otherwise,  $\mathcal{V}' \subset \mathcal{V}$ and $\mathcal{A}' = \mathcal{A}$. Since $\mathcal{G}' = (\mathcal{V}',\mathcal{A}'=\mathcal{A})$ is strongly connected (follows from Lemma~\ref{lem:strong}), $\mathcal{G}$ cannot be strongly connected (contradiction).

Since $\mathcal{G}$ is strongly connected, and $\mathcal{A}'^{\text{c}}$ is non-empty, there must exist an  {\em arc} $(k_1,k_2) \in \mathcal{A}'^{\text{c}} $
such that $k_1\in \mathcal{V}'$ and either  $k_2\in \mathcal{V}'$ or  $k_2\in \mathcal{V}'^{\text{c}}$.
In either case the trail $\mathcal{T}$ on $\mathcal{G}$ can be found in Step~\ref{step:add-path-2}: \\
(i) Suppose $k_2\in \mathcal{V}'$. We have $\mathcal{T} = ( \{k_1,k_2\}, \{ (k_1,k_2)\})$. \\
(ii) Suppose $k_2\in \mathcal{V}'^{\text{c}}$. Since $\mathcal{G}$ is strongly connected, there must exist a path, say $\mathcal{P}'$, from  $k_2$ back to any vertex in $\mathcal{V}'$. Denote the first vertex in $\mathcal{P}'$  that reaches $\mathcal{G}'$ as $k_K\in  \mathcal{V}'$, and the subpath from $k_2$ to $k_K$ as $\mathcal{P}''$. We have $\mathcal{T} =  ( \{k_1,k_2\}, \{ (k_1,k_2)\}) \cup \mathcal{P}''$. Clearly, $k_1, k_K \in \mathcal{V}'$ and the inner vertices are in $\mathcal{V}'^{\text{c}}$, meaning that conditions  in Step~\ref{step:add-path-2} are satisfied with  $\mathcal{T}$ being a path if $k_1\neq k_K$, and being a cycle otherwise.
\end{proof}

\subsubsection{The Reverse SCGC Pruning Algorithm}

Now, for any strongly connected graph, we propose Algorithm~\ref{algo:2} which prunes the graph using the information from the SCGC.
\begin{algorithm}[h]
\ForEach{ trail $\mathcal{T}$ added in Step~\ref{step:add-path-2} of the SCGC}{
remove the first arc of $\mathcal{T}$\;
}

remove any arc from the cycle $\mathcal{C}$ chosen in Step~\ref{step:add-path-1} of the SCGC\;

\caption{The Reverse SCGC Pruning Algorithm}
\label{algo:2}
\end{algorithm}

We have the following results after executing Algorithm~\ref{algo:2}:
\begin{lemma} \label{prop:cycle-tree-free}
Given a strongly connected graph $\mathcal{G}$, after Algorithm~\ref{algo:2} the resulting graph $\mathcal{G}''$ is acyclic with $\maxod=1$.
\end{lemma}

\begin{proof}
We first show that $\mathcal{G}''$ is acyclic, i.e., it does not contain any cycle.
Note that besides the first cycle $\mathcal{C}$ in Step~\ref{step:add-path-1} of SCGC, all other cycles in $\mathcal{G}$ are created in Step~\ref{step:add-path-3} of SCGC. Consider the last appended trail. If we remove the first arc of this trail, we will break all the cycles created by appending this trail. Doing this (i.e., removing the first arc of the appended trail) for the next last-added trail and working backward, we will remove all cycles in $\mathcal{G}$ except $\mathcal{C}$. Now, removing any arc in $\mathcal{C}$ will break the cycle $\mathcal{C}$. So, the resultant graph $\mathcal{G}''$ has no cycle.

Next we show that $\mathcal{G}''$ has $\maxod = 1$. The graph $\mathcal{G}''$ can also be obtained by performing the SCGC and Algorithm~\ref{algo:2} jointly as follows: (i) execute Step~\ref{step:add-path-1} in the SCGC to get $\mathcal{G}'$; (ii) execute Step 2 in the SCGC to get a trail $\mathcal{T}$; (iii) execute Step~\ref{step:add-path-3} in the SCGC but instead of appending the trail $\mathcal{T}$, we append $\mathcal{T}' = \mathcal{T} \setminus (\{k_1\},\{(k_1,k_2)\})$; (iv) remove an arc from the cycle chosen in (i). It is clear that this joint algorithm produces $\mathcal{G}''$. Now, after step (i), $\mathcal{G}'$ is a cycle and hence $\maxod=1$. Next, we consider each iteration (ii)--(iii). We show by induction that if $\mathcal{G}'$ has $\maxod=1$, after appending $\mathcal{T}'$, $(\mathcal{G}' \cup \mathcal{T}')$ has $\maxod=1$.
As $\mathcal{T}$ can only assume a path or a cycle, $\mathcal{T}$ has $\maxod =1$, and so does $\mathcal{T}'$.
Since  $\mathcal{G}' \cap \mathcal{T}' = (\{k_K\}, \emptyset)$, only the vertex $k_K$ can change its outdegree when we append $\mathcal{T}'$ to $\mathcal{G}'$. But the last arc $(k_{K-1},k_K)$ on trail $\mathcal{T}'$ meets the vertex $k_K$ as an incoming arc, so its $\od$ is in fact not changed. Hence, $(\mathcal{G}' \cup \mathcal{T}')$ has $\maxod=1$.
\end{proof}

\begin{lemma} \label{prop:number-of-arcs}
Given a strongly connected graph $\mathcal{G}$, after Algorithm~\ref{algo:2} the number of arcs in the resulting graph $\mathcal{G}''$ equals the number of vertices in $\mathcal{G}''$ less one, i.e., $A(\mathcal{G}'') = V(\mathcal{G}'')-1$.
\end{lemma}

\begin{proof}
From SCGC, after Step~\ref{step:add-path-1}, the number of vertices equals the number of arcs. For each trail with $K$ arcs appended to $\mathcal{G}'$ in Step~\ref{step:add-path-3} of SCGC, we add $K$ new arcs and $K-1$ new vertices to $\mathcal{G}'$. However, because the first arc will be removed in Algorithm~\ref{algo:2}, effectively, the number of vertices added to $\mathcal{G}'$ equals the number of arcs added to $\mathcal{G}'$. In the last step in Algorithm~\ref{algo:2}, an arc is removed from the cycle added in  Step~\ref{step:add-path-1} of SCGC. Hence, the number of arcs in the resulting graph $\mathcal{G}''$ equals the number of vertices in $\mathcal{G}''$ less one.
\end{proof}

\subsubsection{Getting the Lower Bound}

Now, for any single-uniprior problem represented by a strongly connected graph $\mathcal{G} = (\mathcal{V}, \mathcal{A})$, we execute Algorithm~\ref{algo:2} to obtain $\mathcal{G}'' = (\mathcal{V}, \mathcal{A}'')$. We have $\ell^*(\mathcal{G}) \stackrel{(a)}\geq \ell^*(\mathcal{G}'') \stackrel{(b)}= A(\mathcal{G}'')  \stackrel{(c)}=V(\mathcal{G}'')-1 \label{d} \stackrel{(d)}= V(\mathcal{G})-1$,
where (a) is because $\mathcal{A}'' \subseteq \mathcal{A}$, (b) follows from Theorem~\ref{theorem:1} as $\mathcal{G}''$ is acyclic with $\maxod=1$, (c) follows from Lemma~\ref {prop:number-of-arcs}, and (d) is because we only prune arcs but not vertices in  Algorithm~\ref{algo:2}.
 Combining this lower bound with the upper bound in Sec.~\ref{sec:achievability}, we have Theorem~\ref{theorem:strong}. \hfill $\blacksquare$

\begin{remark}
Note that the SCGC and Algorithm~\ref{algo:2} need not be executed for actual coding. They are only used to show that we are able to prune some arcs off a strongly connected graph until the number of arcs equals the number of nodes less one, {\em and} the resultant graph is acyclic with $\maxod=1$.
\end{remark}

\section{General Graphs}\label{section:general}

We now generalize the results in the previous sections to any single-uniprior problem represented by a graph $\mathcal{G}$.

\subsection{The Pruning Algorithm}

To build on the earlier results for specific classes of graphs, we introduce Algorithm~\ref{algo:pruning} that first prunes $\mathcal{G}$ to get $\mathcal{G}'$. It can be shown that Algorithm~\ref{algo:pruning} runs in polynomial time with respect to the number of vertices.

\begin{algorithm}[h]

Initialization: $\mathcal{G}'=(\mathcal{V}',\mathcal{A}') \leftarrow \mathcal{G}=(\mathcal{V},\mathcal{A}) $ \;
1) Iteration: \\
\While{there exists a vertex $i \in \mathcal{V}'$ with\\ $\quad\quad$ (i) more than one outgoing arc, and\\ $\quad\quad$ (ii) an outgoing  arc that does not belong to any\\ $\quad\quad$ cycle [denote any such arc by $(i,j)$]\\}{
remove, from $\mathcal{G}'$, all outgoing arcs of vertex $i$ except for the arc $(i,j)$\;
}
2) label each non-trivial strongly connected component in $\mathcal{G}'$  as $\mathcal{G}_{\text{sub}, i}'$, $i \in \{1,2,\dotsc, N_\text{sub}\}$\;
\caption{The Pruning Algorithm}
\label{algo:pruning}
\end{algorithm}

After Algorithm~\ref{algo:pruning} terminates, we get $\mathcal{G}'=\mathcal{G}_\text{sub}' \cup \left(\mathcal{G}' \setminus \mathcal{G}_\text{sub}'\right)$ where
$\mathcal{G}_\text{sub}' \triangleq  \bigcup_{i=1}^{N_\text{sub}} \mathcal{G}_{\text{sub},i}'$ is a graph consisting of {\em non-trivial}\footnote{A strongly connected component is non-trivial if it has two or more vertices.}  strongly connected components$\{\mathcal{G}_{\text{sub},i}'\}$, and $\mathcal{G}' \setminus \mathcal{G}_\text{sub}' \triangleq (\mathcal{V}' \setminus \mathcal{V}_\text{sub}', \mathcal{A}' \setminus \mathcal{A}_\text{sub}'  )$. Note that $\mathcal{G}' \setminus \mathcal{G}_\text{sub}'$ might not be a graph as there might exist an arc $(i,j) \in \mathcal{G}' \setminus \mathcal{G}_\text{sub}'$ where $i \in \mathcal{G}' \setminus \mathcal{G}_\text{sub}'$ and $j \in \mathcal{G}_\text{sub}'$.
We make the following observations.\\
\indent\emph{Observation 1:} The sets $\{\mathcal{G}_{\text{sub},1}', \mathcal{G}_{\text{sub},2}', \dotsc,\mathcal{G}_{\text{sub},N_\text{sub}}',\mathcal{G}' \setminus \mathcal{G}_\text{sub}'\}$  are vertex and arc disjoint. \\
\indent\emph{Observation 2:} After Step 1 in Algorithm~\ref{algo:pruning}, all vertices with $\od > 1$ in $\mathcal{G}'$ have all outgoing arcs belonging to some cycles. So all these vertices and arcs will eventually be in $\mathcal{G}_\text{sub}'$.

We now state the following main result of this paper: 
\begin{theorem}\label{theorem:general}
For any single-uniprior problem, which can be represented by a graph $\mathcal{G}$, after executing Algorithm~\ref{algo:pruning} we have
\gap
\begin{equation}
\ell^*(\mathcal{G}) = \sum_{i=1}^{N_\text{sub}} (V(\mathcal{G}_{\text{sub},i}')-1) + A(\mathcal{G}' \setminus \mathcal{G}_\text{sub}'). \label{eq:capacity}
\end{equation}
\end{theorem}
Here, $A(\mathcal{G}' \setminus \mathcal{G}_\text{sub}')$ is the number of arcs in $\mathcal{A}' \setminus \mathcal{A}_\text{sub}'$.

\begin{remark}
If $\mathcal{G}$ is acyclic with $\maxod=1$, then $N_\text{sub}=0$ and  $\mathcal{G}=\mathcal{G}' \setminus \mathcal{G}_\text{sub}'$. Thus, we recover Theorem~\ref{theorem:1}.
If $\mathcal{G}$ is strongly connected, then $N_\text{sub}=1$, $\mathcal{G}_{\text{sub},1}'=\mathcal{G}_\text{sub}' = \mathcal{G}'= \mathcal{G}$, and $\mathcal{G}' \setminus \mathcal{G}_\text{sub}'=(\emptyset,\emptyset)$. Thus,  we recover Theorem~\ref{theorem:strong}.
\end{remark}

We prove Theorem~\ref{theorem:general} in the following two subsections.

\subsection{Lower Bound}


Now, for any graph $\mathcal{G}$, we first execute Algorithm~\ref{algo:pruning} to get $\mathcal{G}' =  \bigcup_{i=1}^{N_\text{sub}} \mathcal{G}_{\text{sub},i}' \cup \left(\mathcal{G}' \setminus \mathcal{G}_\text{sub}' \right)$. For each strongly connected components $\mathcal{G}_{\text{sub},i}'$, we execute Algorithm~\ref{algo:2} to get $\mathcal{G}_{\text{sub},i}''$. Let the final graph be $\mathcal{G}'' =  \bigcup_{i=1}^{N_\text{sub}} \mathcal{G}_{\text{sub},i}'' \cup \left(\mathcal{G}' \setminus \mathcal{G}_\text{sub}' \right)$.
Then
\gap
\begin{subequations}
\begin{align}
A(\mathcal{G}'') &= \sum_{i=1}^{N_\text{sub}} (A(\mathcal{G}_{\text{sub},i}'')) + A(\mathcal{G}' \setminus \mathcal{G}_\text{sub}') \label{aa} \\
&= \sum_{i=1}^{N_\text{sub}} (V(\mathcal{G}_{\text{sub},i}'')-1) + A(\mathcal{G}' \setminus \mathcal{G}_\text{sub}') \label{bb} \\
&= \sum_{i=1}^{N_\text{sub}} (V(\mathcal{G}_{\text{sub},i}')-1) + A(\mathcal{G}' \setminus \mathcal{G}_\text{sub}'), \label{cc}
\end{align}
\end{subequations}
where \eqref{aa} follows from Observation 1, \eqref{bb} follows from Lemma~\ref{prop:number-of-arcs}, and \eqref{cc} follows because we prune arcs but not vertices in Algorithm~\ref{algo:2},  so $V(\mathcal{G}_{\text{sub},i}'') = V(\mathcal{G}_{\text{sub},i}')$, for all $i\in\{1,\dotsc, N_{\text{sub}}\}$.

Following from Observation 2, each vertex in $\mathcal{G}' \setminus \mathcal{G}'_\text{sub}$  has $\od \leq 1$.  After executing Algorithm~\ref{algo:2}, from Lemma~\ref{prop:cycle-tree-free}, all $\mathcal{G}_{\text{sub},i}''$ has $\maxod=1$. Hence, $\mathcal{G}''$ has $\maxod=1$.

Any cycle $\mathcal{C} \subseteq \mathcal{G}'$ satisfies  $\mathcal{C} \subseteq\mathcal{G}_{\text{sub},i}'$, for some $i\in\{1,\dotsc, N_{\text{sub}}\}$. Otherwise, the subgraph $\mathcal{G}_{\text{sub},i}'$  that includes only a part of $\mathcal{C}$ is not the {\em maximal} strongly connected subgraph.
After executing Algorithm~\ref{algo:2}, from Lemma~\ref{prop:cycle-tree-free}, all $\mathcal{G}_{\text{sub},i}''$ have no cycle. Thus, $\mathcal{G}''$ has no cycle.

Now, since $\mathcal{G}''$ is acyclic with $\maxod=1$ with $\mathcal{V}'' = \mathcal{V}$ and $\mathcal{A}'' \subseteq \mathcal{A}$, from Lemma~\ref{lem:basic} and Theorem~\ref{theorem:1}, we have  $\ell^*(\mathcal{G}) \geq \ell^*(\mathcal{G}'') = A(\mathcal{G}'') = \sum_{i=1}^{N_\text{sub}} (V(\mathcal{G}_{\text{sub},i}')-1) + A(\mathcal{G}' \setminus \mathcal{G}_\text{sub}')$.  \hfill $\blacksquare$

\subsection{Achievability (Upper Bound)}

We will now show that the number of bits in \eqref{eq:capacity} is indeed achievable. We propose a coding scheme for $\mathcal{G}'$, and then show that the scheme also satisfies $\mathcal{G}$.
\begin{enumerate}
\item For each strongly connected component $\mathcal{G}_{\text{sub},i}'$, we use the coding scheme in Sec.~\ref{section:strong}. Let all the vertices in the subgraph be $\{k_1,k_2,\dotsc, k_{V(\mathcal{G}_{\text{sub},i}')}\}$. The sender transmits the following $(V(\mathcal{G}_{\text{sub},i}')-1)$ bits: $(x_{k_1,k_2}, x_{k_2,k_3}, \dotsc, x_{k_{V(\mathcal{G}_{\text{sub},i}')-1},k_{V(\mathcal{G}_{\text{sub},i})}} )$. Doing this for all $i \in \{1,2,\dotsc N_\text{sub}\}$, the number of bits the sender needs to transmit is $\sum_{i=1}^{N_\text{sub}} (V(\mathcal{G}_{\text{sub},i}')-1)$.
\item For $\mathcal{G}' \setminus \mathcal{G}_\text{sub}'$, we use the coding scheme in Sec.~\ref{section:free}, i.e., the sender transmits $( x_i : \forall (i,j) \in \mathcal{G}' \setminus \mathcal{G}_\text{sub}')$.  The number of bits here is $A(\mathcal{G}' \setminus \mathcal{G}_\text{sub}')$.
\end{enumerate}

From these received bits, all receivers can obtain $\{ x_i : \forall (i,j) \in \mathcal{G}' \setminus \mathcal{G}_\text{sub}'\}$, and each receiver in each $\mathcal{G}_{\text{sub},i}'$ is able to decode $\{ x_j: \forall j \in \mathcal{G}_{\text{sub},i}' \}$. So the transmission requirements of $\mathcal{G}'$ are satisfied.

Now, recall that $\mathcal{G}'$ is obtained after executing Algorithm~\ref{algo:pruning} on $\mathcal{G}$. The only difference in the two graphs is that on the former, some arcs have been removed. However, for any arc removed, say $(i,k)$, the corresponding message is $x_i$, and there exists another arc, $(i,j) \in \mathcal{G}' \setminus \mathcal{G}_\text{sub}'$, $j \neq k$, not removed. In the above coding scheme, $x_i$ is transmitted without coding, and is received by all receivers.  This means the transmission requirements of $\mathcal{G}$ (with additional arcs) are also satisfied.

The total number of bits to be transmitted is $\sum_{i=1}^{N_\text{sub}} (V(\mathcal{G}_{\text{sub},i}')-1) + A(\mathcal{G}' \setminus \mathcal{G}_\text{sub}') \triangleq M$.  So, we have $\ell^*(\mathcal{G}) \leq M$. \hfill $\blacksquare$


\appendix

\section{Proof of Theorem~\ref{theorem:1}} \label{appendix:acyclic}

Here, we prove Theorem~\ref{theorem:1}. We need the following lemma.

\begin{lemma}\label{prop:nodes-indexing}
Every acyclic graph $\mathcal{G}$ has an acyclic ordering of its vertices, i.e.,  $i > j$ for every arc $(z_i,z_j)$ in $\mathcal{G}$.
\end{lemma}
\begin{proof}[Proof of Lemma~\ref{prop:nodes-indexing}]
See proof in \cite[Proposition 1.4.3]{Bang-JensenGutin}  (it is shown instead that  $i < j$ for every arc $(z_i,z_j)$ in $\mathcal{G}$).
\end{proof}

To prove Theorem~\ref{theorem:1}, we start with a lower bound on $\ell^*(\mathcal{G})$.
Denote the set of vertices with outdegree $\od=0$  as $\mathcal{Z} \subset \mathcal{V}$. Using the (re-)indexing method in the proof of Lemma~\ref{prop:nodes-indexing}, we can express $\mathcal{Z} = \{z_1, z_2, \dotsc, z_{|\mathcal{Z}|} \}$. The rest of the vertices is denoted as $\mathcal{Z}^\text{c} = \{z_{|\mathcal{Z}|+1}, z_{|\mathcal{Z}|+2}, \dotsc, z_n\}$. So, the messages to be decoded by at least one receiver are $\{x_{z_{|\mathcal{Z}|+1}}, x_{z_{|\mathcal{Z}|+2}}, \dotsc, x_{z_n}\}$. Since $\maxod=1$, each vertex $z_i \in \mathcal{Z}^\text{c}$ has exactly one outgoing arc, denoted by $(z_i, z_{r(i)})$, i.e., the message $x_{z_i}$ is requested by only one receiver, denoted by receiver $z_{r(i)}$. From Lemma~\ref{prop:nodes-indexing}, we know that $i > r(i)$. Also, the total number of arcs in the graph is $A(\mathcal{G}) =|\mathcal{Z}^\text{c}|=n - |\mathcal{Z}|$.
Let $X_i \in \{0,1\}$ be the random variable for the message bit $x_i$, and denote $X_{\mathcal{S}}\triangleq\{X_i : i\in \mathcal{S}\}$.
Since each message bit is uniformly distributed over $\{0,1\}$, we have that $H(X_i)=1$. Then it follows that
\gap
\begin{subequations}
\begin{align}
&A(\mathcal{G}) = n - |\mathcal{Z}| = \sum_{i \in \mathcal{Z}^\text{c}} H(X_i) = H(X_{\mathcal{Z}^\text{c}} | X_\mathcal{Z}) \label{eq:independent} \\
&= I(X_{\mathcal{Z}^\text{c}}; E(\mathcal{M}) | X_\mathcal{Z}) + H(X_{\mathcal{Z}^\text{c}}|E(\mathcal{M}),X_\mathcal{Z})\\
&= H(E(\mathcal{M}) | X_\mathcal{Z}) - H(E(\mathcal{M}) | X_{\mathcal{V}}) \nonumber\\
&\quad+ \sum_{i =1}^{n-|\mathcal{Z}|} H(X_{z_{|\mathcal{Z}|+i}} | E(\mathcal{M}),X_{z_1}, X_{z_2}, \dotsc, X_{z_{|\mathcal{Z}|+i-1}}) \nonumber\\
&\leq H(E(\mathcal{M})) + \sum_{i =1}^{A(\mathcal{G})} H(X_{z_{|\mathcal{Z}|+i}} | E(\mathcal{M}),X_{z_{r(|\mathcal{Z}|+i)}} ) \label{eq:conditioning}\\
&= H(E(\mathcal{M})) \label{eq:zero-error} \\
&\leq \ell, \label{eq:entropy}
\end{align}
\end{subequations}
where \eqref{eq:independent} follows from the independence of the message bits, \eqref{eq:conditioning} is derived because conditioning cannot increase entropy and by node indexing we have $i > r(i)$ for all $z_i \in \mathcal{Z}^\text{c}$, \eqref{eq:zero-error} follows from the requirement that knowing the sender's messages $E(\mathcal{M})$ and its own message $X_{z_{r(|\mathcal{Z}|+i)}}$, receiver $z_{r(|\mathcal{Z}|+i)}$ must be able to obtain its requested message $X_{z_{|\mathcal{Z}|+i}}$, and \eqref{eq:entropy} is derived because $E(\mathcal{M}) \in \{0,1\}^{\ell}$. Since \eqref{eq:entropy} is true for all $\ell$, we have that $ \min \ell = \ell^*(\mathcal{G}) \geq A(\mathcal{G})$.

Now, we show an index code with $\ell = A(\mathcal{G})$ exists, i.e., $\ell^*(\mathcal{G}) \leq A(\mathcal{G})$. Let $E(\mathcal{M}) = ( x_i: \forall (i,j) \in \mathcal{A} ) \in \{0,1\}^{\ell}$. Since all arcs have distinct tails, we have $\ell = A(\mathcal{G})$. Note that the messages requested by all receivers are captured by the corresponding arcs in $\mathcal{G}$. It follows that having $E(\mathcal{M})$, each receiver is able to obtain its requested message(s). Combining this with the lower bound, we have Theorem~\ref{theorem:1}.
\hfill $\blacksquare$


\end{document}